\documentclass [11pt]{article}
\usepackage{amsmath,amsthm,amsfonts,amscd,eucal,latexsym,amssymb}
\usepackage{epsfig}  % ciao
\def\sp{\hskip -5pt} 
\def\spa{\hskip -3pt}

\def\cH{{\ca H}}
\def\cD{{\ca D}}

\def\cO{{\ca O}}

\def\cA{{\ca A}}

\def\bR{{\mathbb R}}

\newsymbol\rest 1316         %%% restriction symbol 
 
\def\beq{\begin{eqnarray}}
\def\eeq{\end{eqnarray}}
\def\pa{\partial}
\def\at{\left(}               %%  open (
\def\ct{\right)}              %%  close )
\newcommand{\ca}[1]{{\cal #1}}         %%  calligraphic

\def\de{\delta}

\def\ka{\kappa}
\def\la{\lambda}

\def\si{\sigma}
\def\om{\omega}

\def\={\stackrel {\mbox{\scriptsize  def}} {=}} 
\newcounter{remark}[section]

\def\theremark{\thesection.\arabic{remark}}

\def\s #1 {\section{#1}}

\def\ssa #1 {\ifhmode{\par}\fi\refstepcounter{subsection}
  \noindent {\bf\thesubsection}. {\em #1}.\quad
  \addcontentsline{toc}{subsection}{\protect\numberline{\thesubsection} #1}%
  }

\def\ssb #1 {\ifhmode{\par}\fi\refstepcounter{subsection}
  \noindent {\bf\thesubsection.} {\bf #1.}\quad
  \addcontentsline{toc}{subsection}{\protect\numberline{\thesubsection} #1}%
  }

\def\remark {\ifhmode{\par}\fi\refstepcounter{remark}
  \noindent {\bf Remark \theremark}. \quad}

\newtheorem{teorema}{Theorem}[section]
\newtheorem{proposizione}{Proposition}[section]

\newtheorem{definizione}{Definition}[section]

\begin{document} 
 
\hfill{\sl Preprint UTM 742 revised version - March 2011} 
\par 
\bigskip 
\par 
\rm 
 
%%%%%%%%%%%%%   Title %%%%%%%%%%%%%%%%%%%%%%%%%% 
 
\par 
\bigskip 
\LARGE 
\noindent 
{\bf State independence for tunnelling processes through black hole horizons and Hawking radiation} 
\bigskip 
\par 
\rm 
\normalsize 
 
%%%%%%%%%%%%%%%%%%%%%%%%%%%%%%%%%%%%%%%%%%%%%
%%%%%%%%%%%% Authors %%%%%%%%%%%%%%%%%%%%%%%%

\large
\noindent 
{\bf Valter Moretti$^{1,a}$} and  {\bf Nicola Pinamonti$^{2,b}$} \\
\par
\small
\noindent$^1$ Dipartimento di Matematica, Universit\`a di Trento
 and  Istituto Nazionale di Fisica Nucleare -- Gruppo Collegato di Trento, via Sommarive 14  
I-38050 Povo (TN), Italy. \smallskip

\noindent $^2$ Dipartimento di Matematica, Universit\`a di Roma``Tor Vergata", Via della Ricerca Scientifica 1, I-00133 Roma - Italy and
 Dipartimento di Matematica, Universit\`a di Genova, Via Dodecaneso 35, I-16146 Genova - Italy. \smallskip
\smallskip

\noindent E-mail: 
 $^a$moretti@science.unitn.it,  $^b$pinamont@dima.unige.it\\ 
 \normalsize

\par 
 
\rm\normalsize 

%\linespread{1.5} 
\rm\normalsize 

%\linespread{1.5} 
\rm\normalsize 
 
%%%%%%%%%%%% Date %%%%%%%%%%%%%%%%%%%%%%%%%% 
 
\par 
\bigskip 

\noindent 
\small 
{\bf Abstract}.  Tunnelling processes through black hole horizons have recently been investigated in the framework of WKB theory, discovering an interesting interplay  with the Hawking radiation. 
In this paper, we instead adopt the point of view proper of QFT in curved spacetime, namely,
we  use a suitable scaling limit towards a Killing horizon to 
 obtain the leading order of the correlation function relevant for the tunnelling.
 % process through a Killing horizon.  
The computation is done  for certain large class of reference quantum states
  for scalar fields, including Hadamard states.
In the limit of sharp localization either on the external side or on opposite sides of the horizon, 
the quantum correlation functions appear to have thermal nature. In both cases the characteristic temperature is referred to the surface gravity associated with the Killing field and thus connected with the Hawking one. 
Our approach is valid for every stationary charged rotating 
non-extremal black hole. However, since 
the computation is completely local, it
covers the case of a Killing horizon which just temporarily exists in
some finite region, too.
 These results provide strong support to the idea that the Hawking radiation, 
 which is detected at future null infinity and needs some 
 global structures to be defined, 
 is actually related to a 
 local phenomenon taking place even for local geometric structures (local Killing horizons), existing just for a while.  
\normalsize
\bigskip

%\newpage
\s{Introduction }
As is known, the Hawking radiation \cite{Hawking} is detected at {\em future null infinity} of a spacetime containing 
collapsing matter giving rise to a black hole. At least in the case of spherical symmetry, the existence and the features of that radiation are quite independent 
from the details of the collapse. However, the
type of the short-distance behaviour  of the two-point function of the reference state, employed to describe the modes of the radiation, plays a relevant role \cite{FH}.
In recent years, attention has been focused on {\em local} properties of 
%states 
models 
where the Hawking radiation is manifest. Here {\em local} means
 in a neighborhood of a point on the event horizon \cite{PW,fisici,APS,APGS}. In this second approach the radiation appears to be related with some thermal effects 
associated to some tunnelling process through the horizon. In particular, the tunneling probability, computed in the framework of semiclassical WKB approach, 
 has the characteristic thermal form $e^{- E/T_H}$ (the Boltzmann constant 
 being re-defined as $k=1$) where $T_H$ is the Hawking temperature and $E$ is the energy of the particle crossing the horizon. More precisely, that exponential thermal factor arises when taking a limit towards the horizon for an endpoint 
of the path of the classical particle.

This approach is interesting because it deals with local aspects only and, in this sense, it seems to be more general than the standard one. 
In fact,  it may be applied to pictures where a certain geometric structure,  interpreted as the horizon, exists ``just for a while'', without 
extending into a true global structure up to the  future null infinity where, traditionally, the Hawking radiation is detected. 
In \cite{fisici2,fisici3}, even the case of a spherically symmetric black hole {\em in formation} was analysed, where no proper horizon structure exists, 
being replaced by  a {\em dynamical horizon}.
Other interesting results, also considering the backreaction, can be found in \cite{mann,vagenas}.

Within these new remarkable approaches it is however difficult to 
 understand how the found properties are independent form the state of the quantum system. This is
  essentially due to the fact that they are discussed at the {\em quantum-mechanical} level rather than the {\em quantum-field-theory} level.
Indeed, in \cite{PW,fisici} it is assumed that there is some preferred notion of quantum particle whose wavefunction satisfies the Klein-Gordon equation. That equation is actually treated as a Schr\"odinger equation when dealing with transition probabilities within the WKB framework. 
However, in curved spacetime 
 there is no natural definition of quantum particle, unless adopting the  quantum-field-theory framework, fixing a preferred quasifree
  reference state and building up the associated Fock-Hilbert space. It does not seem that 
  this relevant issue is tackled in the mentioned literature.
 Furthermore, the procedure 
exploited in \cite{fisici,fisici2,fisici3} needs a  Feynman prescription to make harmless a divergence that pops up when performing the above-mentioned limit
 towards the horizon. As a matter of fact, that procedure turns the real-axis computation into a 
complex-plane computation and the very imaginary part of the WKB amplitude, arising that way, leads to the wanted factor $e^{- E/T_H}$.  The reason for the appearance of an imaginary part in the semi-classical action
 has been clarified in \cite{APS,APGS}, through a very careful analysis of the integration paths and the adopted coordinate system, without requiring any regularization procedure and remaining in the general WKB framework, while adopting the path integral viewpoint. In any case, that clarification does not solve 
 the problem of the absence of a precise notion of a quantum particle in curved spacetime, necessary to exploit the WKB formalism 
 at quantum-mechanical level. 
 It seems however plausible that the appearance of an imaginary part in the semi-classical action, formally equivalent to the Feynman regularization procedure, is nothing but a remnant of the  choice of a preferred reference state at quantum-field-theory level, whose two-point function
has a short-distance divergence close to that of Minkowski vacuum.

While sticking to the local aspects associated with 
states showing the Hawking radiation, differently from the references quoted above, in this paper we shall deal with
 a definite framework at the quantum-field-theory level. More precisely, we shall focus on the two-point correlation function 
$\omega(\Phi(x)\Phi(y))$ of a quantum field $\Phi$ settled in a (not necessarily {\em quasifree}) state $\omega$
 whose short-distance divergence is, essentially,  of {\em Hadamard type}, thus generalising the 
 short-distance behaviour of Minkowski vacuum. This is one of the  hypotheses exploited in \cite{FH,KW}. It actually 
encompasses a huge class of states, those that are supposed to have a clear physical meaning \cite{Wald2}, especially in relation with the problem 
of the renormalisation of the stress-energy tensor and the computation of the quantum backreaction on the metric.\\
The two-point function $\omega(\Phi(x)\Phi(y))$ corresponds, up to normalization,  to a probability amplitude. 
In this sense,
 it measures the tunnelling probability through the horizon when $x$ and $y$ are kept at the opposite sides of the horizon.  In particular, if the state $\omega$ is  quasifree, that probability amplitude 
can be interpreted as a quantum-mechanical  probability amplitude,  as wished in the above-mentioned literature, just referring to the 
natural notion of particle associated with the state $\omega$.

 {From} the geometric viewpoint, we shall assume to work in a sufficiently small neighbourhood $\cO$ of a local {\em Killing horizon} structure $\cH$, also supposing  that the
 {\em surface gravity} $\kappa$ is {\em nonvanishing} and {\em constant} along the horizon. We stress that the structure could be either part of
 the future Killing event horizon of a stationary black hole in the full Kerr-Newman family also obtained by matter collapse, or it could be completely local and ceasing to exist in the future of $\cO$
 in view of the general dynamics of the matter and the fields in the considered spacetime.
The requirement that the surface gravity is constant on the local horizon means that, at least locally, a thermal equilibrium has been reached, 
since a constant surface gravity corresponds to the validity of the zero-law of black-hole thermodynamics.

The existence of a timelike Killing vector
$K$ defining $\cH$ provides the  preferred notions of time and energy we intend to consider. Notice that in \cite{fisici2,fisici3}, dealing with spherically-symmetric black holes in formation,
 the notion energy was referred to the so called {\em Kodama-Hayward} vector field that, in those backgrounds, generalises the notion of Killing field.

Exploiting  general technical achievements about Killing horizons established in \cite{KW} and \cite{RW1},
 we shall prove that, independently 
from the choice of the quantum state in above-mentioned class, when the supports of the test functions centred on the two arguments
 $x,y$ of $\omega(\Phi(x)\Phi(y))$ become closer and closer to the horizon, the two-point function acquires a thermal spectrum with respect to the notion of time and energy associated with the Killing field. More precisely, if both arguments stay on the same side of the horizon, the Fourier transform of the two-point function presents the very
 Bose-Einstein form driven by the Hawking temperature (see however the remark at the end of the paper).
Conversely, whenever the two arguments are kept at the opposite sides of the horizon, the resulting spectrum is different. It is however in agreement 
with the transition probability between two weakly coupled reservoirs which are in thermal equilibrium at Hawking temperature.
Actually, as in the case of the Boltzmann distribution, its spectrum decays exponentially for  high energies.
%with Boltzmann's distribution  at the Hawking temperature for  high energies. 
In both cases, in order to catch the leading contribution to the two-point function, we shall exploit a suitable scaling limit procedure \cite{HNS,Buchholz,BV} towards
 the horizon. Operating in this way, the local thermal nature of the correlation functions  becomes manifest as a state-independent feature when the states belong to the above-mentioned wide class. 

The paper is organized as follows. In the next section, 
recalling some technical results established in 
\cite{KW} and \cite{RW1},
we shall present the geometric hypotheses we shall use.
We assume that the reader is familiar with 
the standard notions of differential geometry of spacetimes \cite{Wald}. 
%(in the remaining part of the work ``submanifold'' means smooth embedded submanifold). 
 In the subsequent section we shall compute the two-point function $\omega(\Phi(x)\Phi(y))$ and its limit approaching the horizon. The last section will present a summary and some general remarks. 

\s{Spacetime Geometry}
\ssb{Local geometry} We start our discussion fixing the basic geometric setup that we employ in this paper. 
We henceforth consider a $4$-dimensional (smooth) time-oriented spacetime $(M,g)$. 
Furthermore, we  assume the validity of the following local geometric properties, which are the same as in \cite{RW1}.

\begin{definizione}\label{def1}
Let ${\cal O}$ be an open subset of $M$, the 
{\bf local general geometric hypotheses} hold in ${\cal O}$ if a smooth vector field 
$K$ exists thereon such that:
\begin{itemize}
  \setlength{\itemsep}{1pt}
  \setlength{\parskip}{0pt}
  \setlength{\parsep}{0pt}
\item[{\bf (a)}] $K$ is a Killing field for $g$ in $\cO$.

\item[{\bf (b)}] $\cO$ contains a connected  $3$-submanifold ${\cal H}$, the {\bf local Killing horizon},
 that is invariant under the group of local isometries generated by $K$ and  $K^aK_a =0$ on ${\cal H}$.

\item[{\bf (c)}] The orbits of $K$ in ${\cal O}$ are diffeomorphic to an open interval contained in $\bR$ and $\cal H$ admits
 a smooth $2$-dimensional cross section which intersects
 each orbit of $K$ exactly once.

\item[{\bf (d)}]  The {\bf surface gravity} -- {\em i.e.} the function $\kappa : \cH \to \bR$ such that, in view of (a) and (b)  $\nabla^a (K_bK^b)  = -2\kappa K^a$ -- turns out to be strictly positive\footnote{What actually matters is $\kappa \neq 0$, since $\kappa>0$ can always be  obtained in that case by re-defining 
$K \to -K$.} and constant on $\cH$.

\end{itemize}

\end{definizione}
%}\\
 
 \noindent As we said above, the local Killing horizon $\cH$ may represent a horizon which exists ``just for a while'', without 
extending into a true global structure which reaches the  future null infinity.
However, our hypotheses are, in particular, valid  \cite{RW1} in a neighbourhood of any point on a black hole future horizon,
  once that, after the collapse, the metric has settled down to its stationary, not necessarily static,  form 
of any non-extreme black hole in the charged Kerr-Newman family. In particular, those hypotheses and our results are valid for the Kerr black hole.
 There $K$ is the Killing vector defining the natural notion of time in the external region of the black hole and 
 ${\cal H}$ is part of the event horizon.\\
With the hypotheses (a) and (b), the integral lines of $K$ along ${\cal H}$ 
can be re-parametrized to segments of null geodesics and  $\nabla^a (K_bK^b)  = -2\kappa K^a$ holds
 on $\cal H$
where the {\bf surface gravity}, $\kappa : {\cal H} \to \bR$, is constant along each fixed geodesic \cite{Wald}. 
The requirement (d) is not as strong as it may seem at first glance. Indeed, it is possible to prove that, whenever a spacetime admitting a Killing horizon satisfies 
Einstein equations and the dominant energy condition is verified,  the surface gravity must be constant on the horizon \cite{Wald}. (However, \cite{Wald}
 a result originally obtained by Carter  states that -- independently of any field equations -- the 
surface gravity of a Killing horizon must be constant if (i) the 
horizon Killing field is static {\em or} (ii) there is an additional 
Killing field and the two Killing fields are 2-surface orthogonal.)
At least in the case $\kappa > 0$, this is nothing but the
 {\em zero-th law of black hole thermodynamics} where $\frac{\kappa}{2\pi}$ amounts to
the {\em Hawking temperature} of the black hole.\\

\ssb{Killing and Bifurcate Killing horizons}
We  now focus on the relation of the previously introduced local geometric hypotheses and the more rigid case of a bifurcate Killing horizon. 
A Killing field $K$ determines a {\bf bifurcate Killing horizon} \cite{Boyer}
when it vanishes on a connected $2$-dimensional acausal space-like submanifold ${\cal B} \subset M$,
 called the {\bf bifurcation surface}, and $K$ is light-like on the two $K$-invariant $3$-dimensional
 null submanifolds ${\cal H}_+, {\cal H}_- \subset M$ 
generated by the pairs of null geodesic orthogonally emanated from ${\cal B}$. 
 In particular ${\cal H}_+ \cap  {\cal H}_-= {\cal B}$ and the null geodesics forming ${\cal H}_+ \cup  {\cal H}_-$ are re-parametrised integral lines of $K$ on 
$({\cal H}_{+}\cup {\cal H}_-)\setminus {\cal B}$.
By definition, on ${\cal H}_+$ the field $K$ in the future of ${\cal B}$ is outward pointing.  
The simplest example of a bifurcate Killing horizon is that realized by 
the Lorentz boost $K$ in Minkowski spacetime. 
Other, more interesting, cases are the bifurcate Killing horizons of maximally extended
black hole geometries like the Kruskal extension of the Schwarzschild
even including the non-extreme charged rotating case.

For our purposes it is important to notice that, 
in the case of a bifurcate Killing horizon, any neighbourhood ${\cal O}$ of a point on ${\cal H}_+$, which has empty intersection with the bifurcation surface ${\cal B}$, satisfies the 
 local general hypotheses stated in definition \ref{def1}.
It is very remarkable for physical applications and for our subsequent discussion in particular, that such a result can be partially reversed 
as established by Racz and Wald \cite{RW1,RW2}. 
Indeed, if the local general geometric hypotheses are fulfilled  on $\cal H$,
 for a sufficiently small ${\cal O}$,  the spacetime outside ${\cal O}$ can be smoothly deformed preserving the geometry inside ${\cal O}$
 and  extending $K$ and $\cal H$ to a whole bifurcate Killing horizon in the deformed spacetime.
Thus, when studying local properties, the bifurcation surface $\cal B$ can be ``added'' also to those 
spacetimes without bifurcation surface, as is the case of black holes formed by stellar collapse.
One can therefore take advantage of the various technical properties of the bifurcate Killing horizon  as we shall do in the rest of the paper.\\

\begin{figure}
\centering
     \begin{picture}(200,150)(0,0) %(54,90)(0,0)
       \put(0,0){\includegraphics[height=6.5cm]{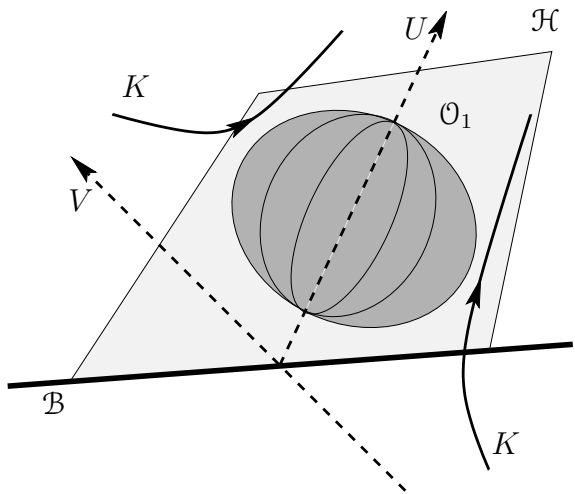}}
       \put(152,172){\large $U$}
       \put(25,107){\large $V$}
       \put(165,140){\large ${\cal O}_1$}
       \put(15,30){\large ${\cal B}$}    
      \put(200,175){\large ${\cal H}$}
      \put(45,150){\large $K$}
       \put(185,15){\large $K$}
 \end{picture}
     \caption{The thick black line is part of the bifurcation surface. 
     The region painted in light grey corresponds to the horizon $\cH$, while the dark-grey region  represents the open set $\cO_1$.}
\end{figure}

\ssb{Killing vector, geodesical distance in $\cO$ and coordinates $V,U,x^3,x^4$} \label{seccoord}
Let us focus on a relevant coordinate patch \cite{KW} defined in a neighbourhood of  ${\cal H}_+$  for a
  bifurcate Killing horizon generated by a Killing vector field $K$ (similar adapted coordinates exist for $\cH_-$). 
Let $U$ denote an affine parameter along the null geodesics forming ${\cal H}_+$ fixing  the origin at the bifurcation surface ${\cal B}$. 
Each point $p\in {\cal H}_+$ is thus determined  by a corresponding pair $(U, s)$, where $s\in {\cal B}$ denotes the point which is
intersected by the null geodesic generator through $p$.

We shall now extend those coordinate system to a neighbourhood of ${\cal H}_+$.
To this end, for each point $q\in {\cal B}$, let us indicate by $n$ the
unique future-pointing null vector which is orthogonal to ${\cal B}$ and has inner product $-1/2$ with $\frac{\partial}{\partial U}$.
We extend $n$ on all of ${\cal H}_+$ by parallel transport along the null generators of ${\cal H}_+$. Let $V$ denote the affine
parameter along the null geodesics determined by $n$, with V= 0 on ${\cal B}$. 
It is clear that $(V, U, s)$ determines a
point in a sufficiently small neighbourhood  of ${\cal H}_+$.  
We are thus in a position to introduce the sought coordinate patch.  
If  $(x^3,x^4)$ denote coordinates defined on a open neighbourhood (in ${\cal B}$) of a point in ${\cal B}$, a coordinate patch 
$(V,U,x^3,x^4)$, we call {\bf adapted to $\cH_+$}, turns out to be defined in corresponding open neighbourhoods %(in $M$) 
of points on 
${\cal H}_+$. 
In these coordinates: 
\beq\label{metric}
g\spa\rest_{{\cal H}_+} = - \frac{1}{2}dU \otimes dV -  \frac{1}{2}dV \otimes dU + \sum_{i,j=3}^4 h_{ij}(x^3, x^4) dx^i\otimes  dx^j \:, \label{met}
\eeq
where metric $h$ is that induced by $g$ on $\cal B$.  $h$ is thus {\em positive definite} and it does not depend on $V,U$.
We stress once more that,  in view of Racz-Wald's result, such a coordinate system always exits in $\cO$, provided the local geometric hypotheses hold in $\cO$, regardless the actual geometry of the spacetime outside $\cO$.

In the rest of the paper, referring to this geometric structure in ${\cal O}$, we shall employ the following notation. 
We shall indicate by ${\cal S}_{V,U}$ the cross section of ${\cal O}$  at $V,U$ constant. Moreover, 
$s(p) \in \cal B$ will be the point with coordinates $(x^3,x^4)$ when 
$p \in \cO$ has coordinates $(V,U,x^3,x^4)$ 
(similarly $(V',U',x'^3,x'^4)$ will denote the analogous set of coordinates for $p'\in {\cal O}$).
The set  $G_\delta(p,V,U) \subset {\cal S}_{V,U}$ is defined as the set whose image $s(G_\delta(p,V,U)) \subset {\cal B}$  coincides with the open $h$-geodesical ball centred on 
$s(p)$ with radius $\delta$.

 We shall denote by $\sigma(p,p')$ the squared $g$-geodesic distance,
 taken with its Lorentzian sign,
% between points $p,p'\in \cO$ (taken with Lorentzian its sign) 
between any couple of points $p,p'$
contained in some $g$-geodesically convex neighbourhood. Finally, we denote by $\ell(s,s')$ 
 the squared $h$-geodesic distance between points $s,s'$ in some $h$-geodesically convex neighbourhood contained in ${\cal B}$.
 
Making use of the preceding definitions and the introduced notations, we are going to
present the following useful proposition, which is also based on achievements in \cite{KW}.

\begin{proposizione}\label{propin} Let ${\cal O}\subset M$ be a set on which the local general geometric hypotheses hold, and let   ${\cal O}_1$ be  another open set such that $\overline{\cal O}_1\subset {\cal O}$ is compact.
Define a coordinate frame adapted to $\cH_+$  on $\cO$ (as said at the beginning of section \ref{seccoord})
 so that $p\in \cO$ has coordinates $(V,U,x^3,x^4)$  with $p \in \cH$ iff $V=0$. The following holds.
\begin{itemize}
  \setlength{\itemsep}{1pt}
  \setlength{\parskip}{0pt}
  \setlength{\parsep}{0pt}

\item[{\bf (a)}] In ${\cal O}$, the decomposition
$K =  K^1 \frac{\partial}{\partial V} + K^2 \frac{\partial}{\partial U} +   K^3 \frac{\partial}{\partial x^3}+
 K^4 \frac{\partial}{\partial x^4}$ is valid and, if $p \in {\cal O}_1$:
\beq\label{K}
K^1(p) = - \kappa V + V^2 R_1(p)\:,  \quad K^2(p) =  \kappa U + V^2 R_2(p)\:, \quad  
K^i(p) =   V R_i(p)\:,\: i=3,4\:, \label{KV}
\eeq 
where $R_1$, $R_2$, $R_i$ are bounded smooth functions defined on ${\cal O}_1$.

\item[{\bf (b)}] If ${\cal O}_1$  is  included in a $g$-geodesically convex neighbourhood and if $p,p' \in \cH \cap {\cal O}_1$, then 
$\sigma(p,p') = \ell(s(p),s(p'))$. 

\item[{\bf (c)}] If ${\cal O}_1$  is as in (b), there exists a $\delta>0$ such that,
 for every fixed $p\in {\cal O}_1$, the smooth map $G_\delta(p,V',U') \ni p' \mapsto \sigma(p,p')$ 
has vanishing gradient (with respect to the coordinates of $s(p')$) in a unique point $q(p,V',U')$ attaining its minimum there. In particular 
$s(q(p,V',U'))= s(p)$ if $p\in \cH$.

\item[{\bf (d)}] For $p$ and $q=q(p,V',U')$ as in (c):
\beq\label{sigma}
\sigma(p,q) =  \ell(s(p),s(q)) - (U-U')(V-V')  + R(p,V',U')\label{expansions}\:,
\eeq 
where 
  $R(p,V',U') = A V^2+BV'^2 + CVV'$,  for some bounded smooth
 functions $A,B,C$ of $p,V',U'$.
 \end{itemize}
\end{proposizione}

\begin{proof}
{\bf (a)}  $\nabla_a K_b + \nabla_bK_a =0$ and  $\nabla^a (K_bK^b)  = -2\kappa K^a$  on $\cal H$ imply that 
$\nabla_K K = \kappa K$ on $\cH$, so that
$K\spa \rest_{{\cal H}} = \kappa U \frac{\partial}{\partial U}$ because $U$ is an affine parameter and $K$ vanishes on ${\cal B}$ where $U=0$.
If $x^a$ is any of $x^1=V,x^2= U,x^3,x^4$,  exploiting the parallel transport used to define the coordinates, we have
\beq 
\Gamma^a_{21}\spa\rest_{\cal H} =  \Gamma^a_{12}\spa\rest_{\cal H} = \Gamma^a_{22}\spa\rest_{\cal H}= \Gamma^a_{11}
\spa\rest_{\cal H} = 
\Gamma^1_{2a}\spa\rest_{\cal H}
= \Gamma^1_{a 2}\spa\rest_{\cal H} 
= \Gamma^2_{a1}\spa\rest_{\cal H} = \Gamma^2_{1a}\spa\rest_{\cal H} =0 \label{idgamma}\:.
\eeq
Above, the fifth one is equivalent to $g(\frac{\partial}{\partial U}, \nabla_{\frac{\partial}{\partial U}}\frac{\partial}{\partial x^a})\spa \rest_{\cal H}=0$ and it arises from
$g(\frac{\partial}{\partial U}, \frac{\partial}{\partial x^a})= -n_{a}$   and 
$\nabla_{\frac{\partial}{\partial U}} \frac{\partial}{\partial U} =0$ on ${\cal H}$, the seventh one can be proved similarly.
Next, taking the first-order Taylor expansion in $V$ of both $K^1$ and $g^{ab}$ about $V=0$, we have, for some smooth function $V^2R_1(p)$  
 bounded in view of the compactness of $\overline{{\cal O}_1}$,
\beq K^1 = K^1\spa\rest_{\cal H} + V \left(g^{1b}\spa\rest_{\cal H} \frac{\partial K_b}{\partial V}\spa\rest_{\cal H} + 
K_b \spa\rest_{\cal H} \frac{\partial g^{1b}}{\partial V}\spa\rest_{\cal H} \right) + V^2 R_1(p)\:.\label{inter}\eeq
{From} $\frac{\partial g^{ab} }{\partial x^c } = - g^{bd} \Gamma^a_{cd} - g^{ad} \Gamma^b_{cd}$ and
and  (\ref{idgamma}), exploiting
 (\ref{met}) and $K^1=0$ on $\cal H$, the identity
 (\ref{inter}) simplifies: The last derivative
 vanishes and
$K^1 = - V \frac{\partial K_2}{\partial V}\spa\rest_{\cal H} + V^2R^1(p)$. 
Furthermore, 
$\nabla_2 K_1 + \nabla_1 K_2 =0$ evaluated on ${\cal H}$ and using (\ref{idgamma})  leads to
$\frac{\partial K_2}{\partial V}\spa\rest_{\cal H} =- \frac{\partial K_1}{\partial U}\spa\rest_{\cal H}$, so that:
$$K^1 = - V \frac{\partial K_1}{\partial U}\spa\rest_{\cal H} + V^2R_1(p) = 
 V \left(g_{1a}\spa\rest_{\cal H}  \frac{\partial K^a}{\partial U}\spa\rest_{\cal H} +
 K^a \spa\rest_{\cal H}  \frac{\partial g_{1a}}{\partial U}\spa\rest_{\cal H}  \right)+ V^2R_1(p)\:.$$
The last derivative vanishes in view of  (\ref{idgamma}), (\ref{met}) and the know identity 
$\frac{\partial g_{ab} }{\partial x^c } =  g_{bd} \Gamma^d_{ca} + g_{ad} \Gamma^d_{cb}$. Therefore the first identity in
(\ref{KV}) holds in view of (\ref{met}) and $K\spa \rest_{{\cal H}_R} = \kappa U \frac{\partial}{\partial U}$. 
The second one can be proved with the same procedure 
noticing that the Killing identity $\nabla_V K_1 = 0$, on $\cal H$,
 becomes $\frac{\partial K_1}{\partial V}\spa\rest_{\cal H}=0$
in view of (\ref{idgamma}). The last identity in (\ref{KV}) is obvious. 

{\bf  (b)} 
 Since the geodesically convex neighbourhoods form a base of the topology and the projection $\pi : p \mapsto s(p)$ is continuous,
if ${\cal O}_1$ is chosen to be sufficiently small, we have
that ${\cal O}_1$ is contained in a geodesically convex neighbourhood
 while, at the same time,  $\pi({\cal O}_1)$ is contained in a $h$-geodesically convex neighbourhood in ${\cal B}$.
 Without loss of generality, we can further assume that the latter is included in a $g$-geodesically convex neighbourhood of $M$.
Thus $\sigma(p,p')$, $\sigma(s(p)),s(p'))$ and $\ell(s(p),s(p'))$ are simultaneously well-defined if $p,p' \in {\cal O}_1\cap {\cal H}$ for a sufficiently small ${\cal O}_1$. 
We notice that, 
$\sigma(p,p')$ is invariant under the action of the  
Killing isometry.
Hence, for any $p,p' \in {\cal H}\cap {\cal O}_1$ we get the 
identity
 $\sigma(p,p')= \sigma(s(p),s(p'))$ 
 taking the limit towards $\cal B$ of the flow 
generated by the Killing field $K$ applied to $p,p'$.  Finally $\sigma(s(p),s(p')) =\ell(s(p),s(p'))$ because ${\cal B}$ is totally geodesic as it can be proved by direct inspection.

{\bf (c)}  Let $(V,U, s) \equiv p$ and $(V',U', s') \equiv
p'$. Whenever both points $p$ and $p'$ are contained on the horizon,
namely $V=V'=0$, 
the thesis holds
in view of (b) and the fact that $\ell(s,s')$ is positive
definite, with positive-definite Hessian matrix in the coordinates
$x'^3,x'^4$ of $s'$.
Furthermore, in this case
 $s(q(p,0,U')) = s(p)$.
By continuity, that Hessian matrix remains positive definite if $p,p'$
stay close to $\cH$, so that, any zero $q(p, V',U')$ of the
$x'^3,x'^4$-gradient of ${\cal S}_{V',U'}\ni p'\mapsto \sigma(p,p')$
determines a minimum of $\sigma(p,p')$.
Taking the Taylor expansion of $\nabla_{x'^i}\sigma(p,p')$ ($i=3,4$)
centred on a point in $\cH\times \cH$
with respect to all the coordinates of $p$ and $p'$,
 the equation for $q(p, V',U')$
can easily be handled by exploiting Banach's fix point theorem, proving
the existence and the uniqueness of $q(p, V',U')$ for $p \in \cO_1$
sufficiently shrunk around $\cH$, and $p'$ varying in a neighbourhood
$G_\delta(p, V',U')$
of $(0,U',s)$ in ${\cal S}_{V',U'}$. We recall that $G_\delta(p, V',U')$ is
the pre-image through ${\cal S}_{V',U'} \ni p' \mapsto s(p')$ of a
geodesic ball on ${\cal B}$ centred on
$s(p)$. The compactness of $\overline{{\cal O}_1}$ and a continuity
argument
assures that $\delta>0$ can be chosen uniformly in $p$.

{\bf (d)} Keeping $U,U',x^3,x^4$ fixed,
the expansion (\ref{expansions})  is nothing but the first-order $(V,V')$-Taylor expansion of  
$\sigma(p, q(p, V', U'))$ at $V=V'=0$, paying attention to the fact that the coordinates $x'^3,x'^4$ of $q(p, V', U')$ depend on $V$ through 
%the dependence of  $q(p, V', U')$ from 
$p$.
\end{proof}

\section{Correlations across the Killing horizon}

\ssb{General outlook} 
We wish to compute the correlation functions of a real scalar quantum field, $\Phi$,  for field observables localized in a region intersecting a Killing horizon. 
Thus, we assume that a quantum system is described by a corresponding
 $^*$-algebra $\cA$ generated by the unit $I$ and field operators $\Phi(f)$ for all $f\in C_0^\infty(M)$.
For our purposes only minimal properties of such a quantum field theory are necessary, namely that the (abstract) field operator $\Phi$ is smeared with  compactly supported smooth functions $f \in C_0^\infty(M)$
 and furthermore that: (i) the map $f \mapsto \Phi(f)$ is linear, (ii) $\Phi(f)^* = \Phi(\overline{f})$,
and (iii) $[\Phi(f),\Phi(g)]=0$ when $supp(f)$ and $supp(g)$ are causally separated. 
 {\em  We stress that do not assume that $\Phi$ satisfies any specific field equation, so that, in principle our approach may encompass interacting fields}.\\
If $\omega$ is an (algebraic) state on $\cA$, 
the {\em correlation functions} we are interested in are the bilinear functionals that map {\em real} smooth  
functions $f,f'$ 
to $\omega(\Phi(f)\Phi(f'))$. We shall specify shortly the form of the test functions $f,f'$. For the moment we only say that
their supports are taken in a open region  $\cO$ containing a Killing field $K$ and which satisfies the local general geometric hypotheses. We finally  assume that $\cO$
can be covered by {\em coordinates adapted to 
$\cH_+$} (defined at the beginning of section \ref{seccoord}) with $U$ and $V$ increasing towards the future.\\
Later  we shall restrict $\cO$ to a subregion $\cO_1$ as in (b) of proposition \ref{propin} because
 we want to use the expression (\ref{expansions}) for the geodesic distance.
The region $\cO_1$ when considered in coordinates $(V,U,s)$, can always be taken of the form $(-a,a)\times (b,c) \times S$, where $S \subset {\cal B}$ is an open relatively compact subset. 
Notice that, shrinking $\cO_1$ around the region of $\cH \cap \cO_1$ 
determined by $(b,c) \times S$ means taking $a>0$ smaller and smaller.

In the sufficiently small neighbourhood $\cO_1$ one finds $g(K,K) = \kappa^2 UV + O(V^2)$ in view of (\ref{met}),  (a) 
in proposition \ref{propin}, and $\partial_Vg_{22}|_\cH=0$ (arising from $\Gamma^2_{22}|_\cH =0$). As a consequence, $K$ turns out to be  spacelike in  ${\cal O}_s\equiv \{p \in {\cal O}_1 \:|\: V(p) >0 \}$ and timelike in  ${\cal O}_t\equiv \{p \in {\cal O}_1 \:|\: V(p) <0 \}$. 
Referring to stationary black holes, ${\cal O}_1$  can be interpreted as a sufficiently small 
region around a point on the future horizon, the only horizon existing when the black hole is produced by collapsing matter. There, ${\cal O}_s$  is part of the {\em internal} region, containing the singularity, while 
${\cal O}_t$ stays in the {\em external} region, stationary with respect to the Killing time associated to $K$. In this way, a notion of {\em energy} 
related to $K$ can be defined, in $\cO_t$ at least, and we will take 
advantage of it shortly.

  We recall that, in the GNS representation of a state $\omega$ on the $^*$-algebra of field observables, 
the expectation value of the product of two fields  $\omega\left( \Phi(f)  \Phi(f') \right)$
(with $f,f'$ real)
is equal to $\langle  \widehat{\Phi}(f) \Psi_\omega|  \widehat{\Phi}(f')   \Psi_\omega   
\rangle$, where $\Psi_\omega$ is the cyclic vector and $\widehat{\Phi}(f) $ is  the  
field operator represented as a proper operator on the GNS Hilbert space. 
Hence, up to normalization, $|\omega\left( \Phi(f)  \Phi(f') \right)|^2$  can be interpreted as a transition probability
between the states $\widehat{\Phi}(f) \Psi_\omega$ and   $\widehat{\Phi}(f') \Psi_\omega$. In this sense,
 when $f$ and $f'$ are localized on the opposite sides 
of the horizon, the regions $\cO_s$ and $\cO_t$, the correlation function $\omega\left( \Phi(f)  \Phi(f') \right)$ provides a measure of the {\em transition probability} through the horizon.\\
Inspired by the ideas proper of the {\em scaling-limit procedure} \cite{HNS,Buchholz,BV},
in order to obtain the leading order to that probability, we shall consider some sequences of smearing functions  $f_\la$ and $f'_\la$ whose support become closer and closer to the horizon $\cH$  in the limit  $\lambda \to 0^+$.
Thus we are going to compute the limit:
$$
\lim_{\la\to 0^+} \omega\left( \Phi(f_\la)  \Phi(f'_\la) \right)  
$$
where  $f_\la$ and $f'_\la$ are smooth functions supported in $\cO_1$
whose supports become closer and closer to the horizon as long as $\lambda \to 0^+$.\\
In contrast to the discussion presented in \cite{HNS}, where the scaling limit {\em towards a point} is analysed, with the proposed procedure we are instead considering scaling limits {\em towards a null surface}, so that the results presented in \cite{HNS} cannot be automatically applied to the present case.\\
Since only the short distance behaviour of the two-point function of the state
is relevant for our computation, we 
select the class of allowed states looking at their ultraviolet features.  
We assume that the two-point function  of  $\omega$ is a distribution of $\cD'(M\times M)$ 
defined as 
$$
\omega\left( \Phi(f)  \Phi(f') \right)  = \lim_{\epsilon\to 0^+} \int_{M\times M} \om_\epsilon (x,x') f(x) f'(x') dx dx'
 $$ 
where the integral kernels  $\om_\epsilon$ have the following form  
\beq
\om_\epsilon(x,x') =\frac{\Delta(x,x')^{1/2}}{4\pi^2\si_\epsilon(x,x')}+ w_\epsilon(x,x')\:, \label{ultra}
\eeq
whenever the test functions are supported  in a fixed, relatively compact, geodesically convex neighbourhood.
In the previous expression,  $\si_\epsilon(x,x')=\si(x,x')+2i\epsilon (T(x)-T(x'))+\epsilon^2$ and $T$ is
a fixed (arbitrarily chosen) time function  \cite{KW}. The smooth strictly-positive function
$\Delta$ is the so-called Van Vleck-Morette determinant \cite{Wald2,KW}.
 We finally assume that $w_\epsilon$ has a ``less singular'' behaviour with respect to 
that of $\si_\epsilon(x,x')^{-1}$ in the sense we are going to specify.
An important case is:
\beq 
w_\epsilon(x,x') = v(x,x') \ln \si_\epsilon(x,x') + w(x,x') \quad \mbox{for some   fixed smooth functions
$v$, $w$.} \label{ultra2}
\eeq 
With this form of $w_\epsilon$, the right-hand side of  (\ref{ultra}) is a straightforward generalization of the short distance 
structure of the two-point function of Minkowski vacuum for a Klein-Gordon scalar field, if one also suppose that 
the field $\Phi$ satisfies the Klein-Gordon field equation 
$\Box \Phi + V \Phi =0$ (where $V: M \to \bR$ is any fixed smooth function).
In particular,  (\ref{ultra})-(\ref{ultra2}) are fulfilled by all the quasifree states  of {\em Hadamard type} \cite{Wald2} that are defined by  (\ref{ultra})-(\ref{ultra2}) with a further requirement  
on the form of $v$ related to the Klein-Gordon equation. Those states are supposed to be the most significant states for free QFT in curved spacetime \cite{Wald2} and are very often employed in the rigorous description of thermal properties of quantum fields 
in the presence of black holes \cite{KW,FH,DMP}. However, since we do not 
 need to consider free fields or any precise field equation, we further relax the requirements, adopting (\ref{ultra})
but, in place of  (\ref{ultra2}), assuming a weaker pair of requirements  (in the following 
$p=(V, U, s)$, $p'=(V', U', s')$): 

{\bf (W1)} $w_\epsilon(p,p') \to w'(p,p') $ as $\epsilon \to 0^+$, almost everywhere in $(p,p')$, for some function $w'$
 and  $w_\epsilon$ is $\epsilon$-uniformly
bounded by a locally $M^2$-integrable function;

{\bf (W2)} $w'(V, U, s, V', U', s') \to w''(U, s, U', s')$ almost everywhere in $(U,s,U',s')$
when  $(V,V')\to (0,0)$ for some function $w''$ on $\cH^2$, and  $w'$ is $(V,V')$-uniformly
bounded by a locally $\cH^2$-integrable function.

 Notice that (W1) and (W2) are satisfied by $w_\epsilon$  whenever it satisfies (\ref{ultra2}) and thus they are valid  for Hadamard states in particular.

To conclude, few words about the precise construction of the functions $f_\lambda,f'_\lambda$ are necessary. 
Let  $f$ and $f'$ be some smooth functions with compact support contained respectively in the regions $\cO_s$ and $\cO_t$. The associated functions $f_\lambda,f'_\lambda$ are defined as 
\beq
f_\la(V,U,x^3,x^4) := \frac{1}{\la}f \at\frac{V}{\la},U,x^3,x^4\ct \;, \quad 
f'_\la(V,U,x^3,x^4) := \frac{1}{\la}f' \at\frac{V}{\la},U,x^3,x^4\ct\:,\quad  \lambda >0\:. \label{f1}
\eeq
As in other scaling-limit procedures,  the pre-factor $\la^{-1}$ is introduced in order to keep the result finite.
Finally, to avoid divergences due to zero-modes
 in the limit $\lambda \to 0^+$, 
%as those modes are  invariant under rescaling of the coordinate $V$, 
we  assume that $f,f'$ are of the form\footnote{Indeed, given $f\in C_0^\infty(\cO_1)$, 
 an $F\in C^\infty_0(\cO_1)$ with   $f = \frac{\pa F}{\pa V}$ exists if and only if 
  $\int_\bR f(V,U,x^3,x^4) dV =0$ on $\cO_1$, namely, if and only if $f(\cdot, U,x^3,x^4)$ has no zero modes referring to the $V$-Fourier transform.}:
 \beq f = \frac{\pa F}{\pa V}\;, \quad  f' = \frac{\pa F'}{\pa V}\:, \quad \mbox{for fixed $F,F' \in C_0^\infty(\cO_1)$.}\label{f2}\eeq 
Alternatively, sticking to generic smooth compactly supported functions $f$ and $f'$, the divergent contribution of zero-modes has to be subtracted at the end  of the computations.
With the assumptions \eqref{f1} and \eqref{f2},
 the $\lambda \to 0^+$ limit  of  $\omega\left( \Phi(f_\la)  \Phi(f'_\la) \right)$ is precisely the notion of {\em scaling limit of $\om\at\pa_V\Phi(x) \pa_V\Phi(y)\ct$ towards the horizon} we shall employ.
The result of such limit represents the first contribution to the sought  {\em transition probability}  in an ideal asymptotic expansion for small $\lambda$.\\

\ssb{Computation}
We are in a position to present the most important result of this paper.
The proof of the following theorem exploits techniques similar to those employed in the appendix B of \cite{KW}. However, our result differs from those presented there because we are not interested in 
studying the restriction of the states to the horizon 
in the symplectic approach, while  we intend to compute the scaling limit of the state in the $4D$ smeared formalism. 
The relevant part for our computation arises from the $(V-V')$ contribution to $\si$ in \eqref{expansions} whereas, in \cite{KW}, only the $(U-U')$ contribution plays a relevant role.
In this respect, there are some similarities to the analysis performed in \cite{FH}, although here neither we restrict ourself to the spherically symmetric case nor we consider any equation of motion for the quantum field, nor we suppose that the Killing structure extends up to the future null infinity, as done in \cite{FH}.

\begin{teorema}\label{theorem}
Assume that the general local geometric hypotheses as in Def.\ref{def1} hold 
in $\cO$ (covered by coordinates adapted to $\cH_+$), suppose that $\cO_1 \subset \cO$ 
is a sufficiently small open neighbourhood of a point on $\cH$ with $\overline{\cO_1}\subset \cO$ compact. 
Assume also that the state $\omega$ 
has two-point function given by a distribution as in (\ref{ultra})
 that verifies the requirements (W1) and (W2) above (in particular, $w_\epsilon$ may have the form (\ref{ultra2}) and may be a Hadamard state for Klein-Gordon fields).\\
 If $f,f'$ are taken as in (\ref{f1})-(\ref{f2}) and 
 $\mu$ is the measure associated to the $2$-metric  on the bifurcation surface $\cal B$, it holds
\beq
\lim_{\la\to 0^+} \omega\left(\Phi(f_\la)  \Phi(f'_\la)\right)  = \lim_{\epsilon\to0^+}  -\frac{1}{16\pi} \int_{\bR^4\times {\cal B}}  \frac{F(V,U,s) F'(V',U',s)}{(V- V'-i\epsilon)^2}  dU dV dU' dV' d\mu(s)\:.
\label{correlations}
\eeq
\end{teorema}

\begin{proof}
We start by focusing on the contribution of the most singular part of the two-point function (\ref{ultra}), 
that is the iterated limit:
\beq
L:=  \lim_{\la\to 0^+ } \lim_{\epsilon\to0^+}  \int_{\overline{\cO_1}\times \overline{\cO_1}}   
\frac{\Delta^{1/2}(p,p') f_\la(p) f'_\la(p') }{4\pi^2 \si_\epsilon(p,p')} \; dpdp'\:, \label{intint}
\eeq
where $dp$ is a short-cut 
for the measure induced by the metric.
Fixing $\delta>0$ and $p,V',U'$,  consider a neighbourhood $G_\delta(p, V',U') \subset {\cal S}_{V',U'}$ as in (c) of proposition \ref{propin}, and
define a smooth map ${\cal B} \ni s \mapsto \chi_\delta(s,V',U',p) \geq 0$ with support completely included in $G_\delta(p, V',U')$
and $\chi_\delta(s',V',U',p)=1$ for $0 \leq \sqrt{\lambda(s(p), s')} \leq \frac{\delta}{2} +\frac{1}{2} \sqrt{\lambda(s(p), s(q(p,V',U')))}$. In view of the smoothness of all considered functions  it is possible to arrange these 
functions in order that 
$(s',V',U',p) \to \chi_\delta(s',V',U',p)$ is jointly smooth. Finally we can decompose the integral in (\ref{intint}) as:
\begin{gather} \int_{\overline{\cO_1}} dp \left( \int_{\overline{\cO_1}} \frac{\Delta^{1/2}(p,p') f_\la(p) f'_\la(p') }{4\pi^2 \si_\epsilon(p,p')} \chi_\delta(s',V',U',p) \: \sqrt{|\det g(p')|} dV'dU' ds' \right) \nonumber\\
+ \int_{\overline{\cO_1}} dp \left( \int_{\overline{\cO_1}} \frac{\Delta^{1/2}(p,p') f_\la(p) f'_\la(p') }{4\pi^2 \si_\epsilon(p,p')} (1-\chi_\delta(s',V',U',p)) \: \sqrt{|\det g(p')|} dV'dU' ds' \right) \label{intint2}\:,
\end{gather}
where $ds' = dx'^3 dx'^4$. 
Let us start from the second integral.
As consequence of the compactness of $\overline{\cO_1}$, the continuity of $\sigma$ and (b) in proposition \ref{propin}, 
we have 
 that, 
 for a fixed $\eta >0$ and for $p,p'$ in a sufficiently small $\cO_1$,
 $\sqrt{\sigma(p,p')} \geq \eta/2$ if 
$\sqrt{\ell(s(p),s(p'))} > \eta$ .
We stress that the limit in $\lambda \to 0^+$ in (\ref{intint}) allows us to take $\cO_1$ as
small as we need.
Thus, by definition of $\chi_\de$ and of $G_\delta(p, V',U')$, the denominator
$\si_{\epsilon=0}(p,p') \geq \delta^2/16$ when $1-\chi_\delta \neq 0$.
Hence,
the integrand of the second integral in 
(\ref{intint2}) is  jointly smooth in all variables  including $\epsilon$, even  for $\epsilon =0$. 
Then, in view of Lebesgue's dominated convergence theorem, the limit in $\epsilon$ can be computed simply taking $\epsilon=0$ in the integrand.
For the same reason, after changing the integration variables $(V,V')$ to $(\la V,\la V')$, 
the subsequent limit as $\lambda \to 0^+$ can be 
computed under the sign of integration.
The resulting integral vanish because, in the $(V,V')$ variables, it is noting but the integral of the $V,V'$ derivative of some compactly supported smooth function.
Thus only the former integral in (\ref{intint2}) may survive the limits in (\ref{intint}). Let us focus on that integral.
Making use of (c) in proposition \ref{propin}, in each set $G_\delta(p,V',U') \subset {\cal S}_{V',U'}$ we define the function $\rho(p') \geq 0$ such that:
$$
\sigma(p, p') = \rho(p')^2 + \sigma(p, q(p,V',U'))\:.
$$
In view of (c) in proposition \ref{propin}, the pair $\rho,\theta$, where $\theta\in (-\pi,\pi)$ is the standard polar angle in geodesic polar coordinates centred
 on $q(p, V',U')$, determines
 an allowable local chart for any $p' \in G_\delta(p,V',U')$ (see also the appendix B of \cite{KW}), that is smooth barring the usual conical singularity for $\rho=0$.
 Notice that, due to the last statement in (c) of proposition \ref{propin}, when $p\in \cH$ and $V'=0$, $\rho$ coincides with the standard geodesic radial coordinate centred on $s(p)\in {\cal B}$.
 In the following we shall employ that coordinate system in each $G_\delta(p,V',U')$. Making finally use of (d) in proposition \ref{propin} and choosing $T=(U+V)/2$, we can re-arrange the former
  integral in (\ref{intint2}) so that:
\beq
L  =  \lim_{\la\to 0 } \lim_{\epsilon\to0^+}  \int  \sp 
\frac{\Delta'^{1/2}(p,p')  f_\la(p) f'_\la(p') }{\rho^2 -(V-V'-i\epsilon)(U-U'-i\epsilon)  + R(p,V',U')}  \; \frac{dp' dp}{4\pi^2}\;,\label{uno}
\eeq
where, $\Delta'^{1/2}(p,p')  := \Delta^{1/2}(p,p') \chi_\delta(p,p')$. 
{From} now on we shall assume that the integral in $p'$ is performed before that in $p$. 
Using this coordinate system the integral in the right-hand side of (\ref{uno}) can be rewritten\footnote{We henceforth assume to cut the complex plane along the negative real axis to define the 
function $\ln z$.} as
\begin{gather*}
\int  \sp  \Delta'^{1/2}(p,p') 
f_\la(p) f'_\la(p')   
\frac{\pa}{\pa \rho} 
\ln \at \rho^2 + \sigma(p, q(p,V',U'))\ct 
\frac{\sqrt{|\det g|}}{8\pi^2 \rho}{d \rho d\theta} dU'dV'  dp\nonumber
\end{gather*}
where $\det g$ is the determinant of the metric in the coordinates $\rho, \theta, V',U'$, which parametrically depends on $p$. Notice that the domain of integration in $\rho$
is bounded by the support of the function $\chi_\delta(p,p')$ embodied in $\Delta'$.
For $V=V'=0$, the metric takes the form (\ref{met}) on $\cal B$ which does  not depend on $U,U',V,V'$ any more while $R$ vanishes. 
 By direct inspection one sees that $ \frac{\sqrt{|\det g|}}{\rho}$ is continuous,  
 tends to $1/2$ when $p\to p'$ in $\cH$ and its $\rho$-derivative
  is continuous for $\rho\neq 0$ it being however 
bounded there.
If  $\Delta'_\lambda$,
$R_\la$, $\det g_\lambda$, $dp_\lambda$ are respectively defined as $\Delta'$, $R$, $\det g$ and $dp$ with $V$ and $V'$ rescaled by $\la$, 
changing coordinates $(V,V')\to (\la V,\la V')$ the integral in the right-hand side of (\ref{uno}) can be rearranged as
\begin{gather*} \int \sp \pa_VF(p) \pa_{V'} F'(p') \Delta'^{1/2}_\lambda(p,p')
\frac{\pa}{\pa \rho} \ln\spa \left[\rho^2  - (\la V-\la V'-i\epsilon)(U-U'-i\epsilon) + R_\la(p,V',U')\right] \nonumber\\
 \frac{\sqrt{|\det g_\lambda|}}{8\pi^2\rho} d\rho d\theta dU'dV'  dp_\lambda\;.
\end{gather*}
 Leaving unchanged the remaining integrations, we can first integrate by parts in the polar coordinate $\rho$. We obtain two boundary terms, which are integrals in the 
 remaining variables evaluated respectively  at $\rho=0$ and $\rho=\rho_0$ sufficiently large,
  and an integral in all the variable including $\rho$.
 When taking the limits as $\epsilon \to 0^+$ and $\lambda \to 0^+$, concerning the integral representing the boundary term at $\rho= \rho_0$, we can pass both limits under the sign of the integration by straightforward application of Lebesgue's 
 dominated convergence theorem.
 The result is that, in the integrand, only $\pa_V F$ and $\pa_{V'} F'$ depend on  $V,V'$, hence, performing the integrations 
in $V$ and $V'$ both integrals vanish because $F$ and
 $F'$ are of compact support.
The remaining boundary term leads  to the limit:
\begin{gather*}
L=  \lim_{\la\to 0^+} \lim_{\epsilon\to0^+}  -2\pi \int   \Delta'^{1/2}(\lambda V,U,s,\lambda V',U',s)
\pa_VF(V,U,s) \pa_{V'}F'(V',U',s)
\\
 \at   
\ln \left( -(V- V'-i\epsilon)(U-U'-i\lambda\epsilon) + \frac{R_\lambda(p,V',U')}{\la}  \right) + \ln \la
\ct
 \left.\frac{\sqrt{|\det g_\lambda|}}{8\pi^2\rho}\right|_{\rho=0} dU'dV'  dp_\lambda\;.
\end{gather*}
Above, the factor $2\pi$ arises by the integration in $\theta$ at $\rho=0$,
 and we have safely replaced $\epsilon$ with $\la\epsilon$ in the integral in view of the order of the limits.
Notice that, thanks to (d) of proposition \ref{propin}, 
$|\la^{-1} R_\lambda(p,V',U')| < C \la$ for some constant $C$, when $p,p' \in \cO_1$ and $\lambda \in [0, \lambda_0)$.
Furthermore, from the (\ref{met}) of the metric on $\cH$
we get that 
$dp_\lambda = \frac{1}{2}(1+ \lambda V z) dUdV d\mu(s)$ for some smooth function $z=z(V,U,s)$,
where $d\mu$ is the measure associated to the $2$-metric $h$ on ${\cal B}$. 
Thus the term proportional to $\log(\la)$
can be dropped as it  gives no contribution to the final result 
because $g_0$, $\Delta_0$ do not depend on $V$ and the apparently divergent term
results once more in the integral of a derivative of a compactly supported smooth function. 
The limits of the remaining terms can then be computed, in the given order, 
exploiting Lebesgue's theorem. We eventually obtain:
\begin{gather}
L= - \int   
\pa_VF \pa_{V'}F' \left(i\pi \chi_{E_+}\chi_{A_+} - i\pi \chi_{E_+}\chi_{A_-}    +
\ln |V-V'||U-U'|\right) \frac{dU dV dU' dV' d\mu(s)}{16 \pi}
\;,\label{agg}
\end{gather}
where $E_{\pm}$ is the subset of $\cO_1$ with, respectively, $(V-V')(U-U')\gtrless 0$ and $A_\pm$ is the analogue with, respectively,
$U-U' \gtrless 0$, and $\chi_S$ is the characteristic function of the set $S$. 
We have also used the fact that $dp_\lambda$ becomes $\frac{1}{2}dUdV d\mu(s)$ for $\lambda=0$ in view of the form (\ref{met}) of the metric on $\cH$
and that $\Delta'=\Delta =1$ when $s=s'$ and $V=V'=0$ (it follows from $\Delta(p,p)=1$ and, since $\Delta$ is invariant under isometries, using an argument
 similar to that employed to prove
(b) of proposition \ref{propin}). For the same reason, in view of the meaning $\rho$, 
 $\frac{\sqrt{\det|g_\lambda|}}{\rho} \to 1/2$ for  $\rho \to 0$ when $\lambda =0$, when working in coordinates $\rho,\theta, V,U$.
 The integral in (\ref{agg}) can equivalently be re-written introducing another $\epsilon$-prescription as:
\begin{gather*}
 L= \lim_{\epsilon\to0^+}  -\frac{1}{16\pi} \int   
\pa_VF \pa_{V'}F'
\ln( -(V- V'-i\epsilon)(U-U') ) 
 dU dV dU' dV' d\mu(s)\;.
\end{gather*}
Eventually, integrating by parts, we obtain \eqref{correlations} concluding the proof, provided the contribution of the term $w_\epsilon$ in (\ref{ultra}) to the left-hand side of \eqref{correlations} vanishes.
Indeed this is the case. Taking the requirements (W1) and (W2) into account, applying Lebesgue's dominated convergence theorem, changing variables $V \to \lambda V\:, V' \to \lambda V'$ and exploiting Lebesgue's theorem  again,  we have:
$$ \lim_{\la\to 0^+ } \lim_{\epsilon\to0^+}  \int 
 f_\la(p) f'_\la(p') w_\epsilon(p,p') \; dpdp' =
 \int w''(U, s, U',s') \frac{\partial F}{\partial V}  \frac{\partial F'}{\partial V'} dUdV d\mu(s) 
 dU'dV' d\mu(s') =0
 $$
 in view of the fact that the result of the integrations in $V$ and $V'$ vanishes because $F$ and $F'$ are smooth with compact support.
\end{proof}

\ssb{The correlation functions and their  thermal spectrum}
As is known (e.g., see \cite{Wald}) a timelike Killing vector field on the one hand  provides   a natural notion of time, which is nothing but the parameter of the integral lines of the field. On the other hand
 it gives
a natural notion of conserved energy for fields and  matter propagating in the region where the Killing vector is present. 
We are interested in computing the energy spectrum of the correlation functions $\omega(\Phi(f_\la)\Phi(f'_\la))$ 
%(in the limit $\lambda \to 0^+$)
seen by an observer that moves along the curves generated by the Killing field $K$
 and computed with respect to the associated Killing time.   
 More precisely, exploiting Theorem \ref{theorem}, we intend to compute that energy spectrum in the limit of test functions squeezed on the local Killing  horizon. 
 As the supports of the test functions are infinitesimally close to the horizon,
 we have to focus on what happens for $V\sim 0$. Therefore  we truncate every component of  the formula (\ref{KV}) for the Killing vector at the dominant order in powers of $V$
 and we make use of the right-hand side of (\ref{correlations}) as definition of correlation two-point function.
Now $\tau$ is the Killing time, namely the integral parameter of the curves tangent to $K$. In the said approximation,  the first 
identity in the right-hand side of (\ref{KV}) implies:
\beq 
V(\tau) =  -e^{-\ka \tau} \quad \mbox{for $V<0$ (that is in $\cO_t$) and}\quad V(\tau) = e^{-\kappa \tau}\quad \mbox{for $V>0$ (that is in $\cO_s$),} \label{casiV}
\eeq
up to an additive constant in the definition of $\tau$ which in principle could  depend on the integral curve. 
Our choice is coherent with the standard definitions of $\tau$ in
 Schwarzschild spacetime where $\tau$ is the Killing time in
the external region. Indeed we recover those cases in the limit
 where the $\tau$-constant $2$-surfaces are close to the Killing horizon.
We now examine two situations.\\

{\bf (a)} {\em Both the supports of $f_\la$ and $f'_\la$ stay in $\cO_t$}.\\
In that case, thinking of the functions $F,F'$ as functions of $\tau,\tau'$ instead of $V,V'$ in view of (\ref{casiV}), we can re-arrange the found expression for the correlation function as
\beq
\lim_{\la\to 0 }   \omega(\Phi(f_\la)\Phi(f'_\la)) =
\lim_{\epsilon\to 0^+} - \frac{\kappa^2}{64\pi}   \int \frac{F(\tau,U,x) F'(\tau',U',x)}{(\sinh(\frac{\ka}{2} (\tau-\tau'))+i\epsilon)^2}  d\tau dU d\tau' dU' d\mu(x) \;,
\eeq
where we have used the fact that
the functions $F$ and $F'$ are compactly supported by construction even adopting the new coordinate frame.
It is known that, in the sense of  the Fourier transform of the distributions, 
$\int_\bR \frac{d\tau}{\sqrt{2\pi}}  \frac{e^{-iE \tau}}{(\sinh(\frac{\tau}{2})+i0^+)^2} = -\sqrt{2\pi}\frac{E e^{\pi E}}{e^{\pi E} - e^{-\pi E}}$
(e.g., see the appendix of \cite{DMP}). 
That identity and the convolution 
theorem lead to
\beq
\lim_{\la\to 0 }  \omega(\Phi(f_\la)\Phi(f'_\la)) = \frac{1}{32} \int_{\bR^2\times {\cal B}} \left(\int_{-\infty}^\infty 
 \frac{\overline{\hat{F}(E,U,x)}  \hat{F'}(E,U',x) }{1-e^{-\beta_H E}} E dE \right) dU dU' 
d\mu(x) \;,
\eeq
where $\beta_H={2\pi}/{\ka}$ is the {\em inverse Hawking temperature} and we have defined the smooth function:
$$  \hat{F}(E,U,x) := \int_\bR \frac{d\tau}{\sqrt{2\pi}} e^{-iE \tau} F(\tau,U,x)\:,$$
that, for $E \to \pm \infty$, vanishes faster than $E^{-n}$, $n=0,1,2,\ldots$ uniformly in the remaining coordinates. 
The thermal content of the found correlation function is manifest in view of the Bose factor $({1-e^{-\beta_H E}})^{-1}$ where the Hawking temperature $1/\beta_H$ takes place
 - fixed with respect to a (generally arbitrary) choice of the scale  necessary to define the Killing field $K$ (see the remark below).\\

{\bf (b)} {\em The support of $f_\la$ stays in $\cO_s$, while that of $f'_\la$ stays in $\cO_t$: Tunnelling processes}.\\
As previously remarked, up to normalization, $|\omega(\Phi(f_\la)\Phi(f'_\la))|^2$ can be interpreted as a tunnelling probability through the horizon.
Employing (\ref{casiV}) once more, we end up with:
\beq
\lim_{\la\to 0 }   \omega(\Phi(f_\la)\Phi(f'_\la)) =
\lim_{\epsilon\to0^+} \frac{\kappa^2}{64\pi}   \int  \frac{F(\tau,U,x) F'(\tau',U',x')}{\cosh(\frac{\ka}{2} (\tau- \tau')+ i\epsilon)^2}  d\tau dU d\tau' dU' d\mu(x)\;.
\eeq
As expected from the fact that, in this case, the support of $f_\la$ is always disjoint from the support of $f'_\la$, we can directly pass the limit $\epsilon\to0^+$ 
under the sign of integration simply dropping $i\epsilon$ in the denominator. Taking advantage of the convolution theorem, the final result reads: 
\beq
\lim_{\la\to 0 }   \omega(\Phi(f_\la)\Phi(f'_\la)) = \frac{1}{16} \int_{\bR^2\times {\cal B}} \left(\int_{-\infty}^\infty 
 \frac{\overline{\hat{F}(E,U,x)}  \hat{F'}(E,U',x)}{\sinh(\beta_H E/2)}  EdE \right) dU dU' 
d\mu(x) \:.\label{tunnel} 
\eeq
Notice that, if the arbitrary additive constant defining $\tau$ in $\cO_t$ were different from that in $\cO_s$, then a further exponential $\exp\left(i c E\right)$ would take place in the numerator for some 
real constant $c$.
In any case, the energy spectrum does not agree with the Bose law. However, considering packets concentrated 
around to a high value of the energy $E_0$, (\ref{tunnel}) leads to the estimate for the tunnelling probability:
$$
\lim_{\la\to 0}|\omega(\Phi(f_\la)\Phi(f'_\la))|^2 \sim \mbox{const.}\:  E_0^2 \: e^{-\beta_H E_0}\:,
$$
in agreement with the ideas in \cite{PW,fisici}. It is nevertheless worth remarking that the interpretation of $E$ as an energy is questionable for the packet in the internal region 
$\cO_s$ since
the Killing vector $K$ is spacelike therein.\\

\noindent {\bf Remark}\footnote{The authors are grateful to one of the referees for this important remark.}. It is worth stressing that if $K$ is a (future-directed) Killing field, for every constant $c>0$, $c K$ is a (future-directed) Killing field again. In general, there is no way to fix the scale $c$.  
This arbitrariness in the definition of $K$ enters the definition 
$\nabla^a (K_bK^b)  = -2\kappa K^a$ of the associated surface gravity $\kappa$
and, in turn, it affects the definition of the Hawking temperature $T_H = \kappa/(2\pi)$.
 In the case 
of an asymptotically flat (say, Schwarzschild) black hole, one uses the 
asymptotic behaviour of the Killing isometries to normalize the 
``Schwarzschild time'' parameter to correspond to ``ordinary 
time" at infinity (and of course energy is also defined with 
respect to infinity), fixing the scale $c$.  However, if one is assuming only the 
existence of a local horizon and not necessarily any connection 
with any ``infinity'' region, in the absence of  further 
physical requirement,
there is 
nothing to set the scale. So, while it may be correct to say 
that one has a thermal spectrum in the found expressions for 
$\lim_{\la\to 0}\omega(\Phi(f_\la)\Phi(f'_\la))$, 
it does not seem very  meaningful to claim that this corresponds to any particular 
value for the Hawking temperature in the general case.

\section{Conclusions}
In this paper we have analysed the correlation functions and the tunnelling amplitude through a Killing horizon for quantum  states in a certain, physically relevant, class of states that includes the states  of  Hadamard type. Although the equation of motion does not play a relevant role, the results are in particular valid
 for scalar Klein-Gordon particles.
The computation is performed in the limit of test functions squeezed on the Killing horizon.
The considered local Killing horizon with positive constant surface gravity may be a part of the complete horizon of a black hole, including non-static black holes 
as the non-extremal charged rotating one,
 or it may just temporarily 
exist in a finite region.
The considered states are generally not required to be invariant with respect to the isometry group generated by the Killing field.
We have established  that, in the limit of test functions sharply localized on the opposite sides of the horizon, the correlation functions have
a thermal nature, namely they have a spectrum which decays exponentially as $\exp\{-\beta_{\scriptsize \mbox{Hawking}} E\}$ for high energies. 
The energy $E$ and the inverse Hawking temperature 
 are defined with respect to the Killing
 field generating the horizon (see however the remark at the end of the previous section). 
 This achievement  is in agreement with the result obtained in other recent papers, although here it is obtained in 
 the framework of the rigorous formulation of quantum field theory on curved spacetime. While, 
 in the mentioned literature quantum-mechanical approaches are exploited, leaving unresolved
the issue concerning the strongly ambiguous notion of quantum particle in curved spacetime.

 Furthermore,  we have also established that, when both test functions  are  localized in the external side of the horizon, a full Bose spectrum 
 at the Hawking temperature arises  in the expression of the correlation functions.
 %, when taking  the considered limit. 
 In both
 cases the computation is completely local, i.e. the nature of the  geometry at infinity does not matter and the results do not depend on 
 the employed  states provided they belong to the mentioned class. 
These results give a strong support to the idea that the Hawking radiation, 
 that it is usually presented as a radiation detected at future (lightlike) infinity and needs the global structure  of a black hole Killing horizon, 
 can also be described as a local phenomenon for local geometric structures (local Killing horizons) existing just ``for a while''.  
 A fundamental ingredient in our computation is the fact that 
  the nonvanishing  surface gravity is constant on the Killing horizon. Thus, we could exploit the result in \cite{RW1} and, in turn,
 some technical constructions of \cite{KW}.
 Even if this requirement can  easily be 
 physically interpreted as the geometrical description of the thermodynamic equilibrium, it would be interesting to consider, from our viewpoint, the case  of a black hole 
  in formation, where there are no Killing horizons at all.  
 The latter situation has already been investigated, at least in the presence of spherical symmetry, as   in 
 \cite{fisici2,fisici3}. In those papers the WKB approach as well as the theory of Kodama-Hayward and the associated notion of dynamical horizon are exploited. 
 As a preliminary comment, we notice that the computation of the scaling limit towards the horizon does not seem to require the presence of a proper Killing horizon, which could be replaced by some more generic null hypersurface.

\section*{Acknowledgements.} 
We would like to thank Luciano Vanzo and Sergio Zerbini for useful discussions, suggestions and comments on the subject of this paper. We are also grateful to the referees of CMP for helpful remarks and suggestions.
%The financial support by the German DFG Research Program SFB 676 is kindly acknowledged.
The work of N.P. is supported in part by the ERC Advanced Grant 227458 OACFT ``Operator Algebras and Conformal Field Theory", and in part 
by a grant from GNFM-INdAM under the project
``Stati quantistici di Hadamard e radiazione di Hawking da buchi neri rotanti''.

\end{document}